\documentclass[submission,copyright,creativecommons]{eptcs}
\providecommand{\event}{FVAV 2017}

\usepackage{amsthm}
\usepackage{graphicx} 
\usepackage{color}
\usepackage{amsmath}
\usepackage{tikz}
\usepackage{listings}
\usepackage{pifont}
\usepackage{prettyref}
\usepackage{multirow}
\usepackage{comment}

\def\R{\mathbb{R}}
\newcommand\reals{\mathbb{R}}

\newcommand\vx{\vec{x}}

\newcommand{\vr}[1]{\boldsymbol{#1}}
\newcommand{\vrx}{\vr{x}}
\newcommand{\ddl}{$d{\cal L}$}
\newcommand{\keyx}{\texttt{Keymaera X}}

\newcommand{\cphi}{\varphi}
\newcommand{\blue}[1]{\textcolor{blue}{#1}}

\newcommand{\lie}{{\cal L}}
\newcommand{\Vcst}{V_{\mathit{cst}}}
\newcommand{\dVcst}{\dot{V}_{\mathit{cst}}}
\newcommand{\invcst}{\Vcst \leq 0}
\newcommand{\Vpro}{V_{\mathit{pro}}}
\usepackage{amsmath,amssymb}

\newcommand{\rref}[2][]{\prettyref{#2}}
\newrefformat{sec}{Section\,\ref{#1}}
\newrefformat{def}{Def.\,\ref{#1}}
\newrefformat{defn}{Def.\,\ref{#1}}
\newrefformat{thm}{Theorem\,\ref{#1}}
\newrefformat{prop}{Prop.\,\ref{#1}}
\newrefformat{lem}{Lem.\,\ref{#1}}
\newrefformat{cor}{Corollary\,\ref{#1}}
\newrefformat{ex}{Ex.\,\ref{#1}}
\newrefformat{ax}{Appendix\,\ref{#1}}
\newrefformat{tab}{Table\,\ref{#1}}
\newrefformat{fig}{Fig.\,\ref{#1}}
\newrefformat{alg}{Algorithm\,\ref{#1}}
\newrefformat{exec}{Exec.\,Sch.\,\ref{#1}}
\newrefformat{op}{Op.\,\ref{#1}}
\newrefformat{rk}{Remark\,\ref{#1}}
\newrefformat{pr}{Principle\,\ref{#1}}
\newrefformat{eq}{Eq.\,\eqref{#1}}
\newrefformat{sys}{System\,\eqref{#1}}

\def\vec#1{\mathchoice{\mbox{$\mathbf\displaystyle#1$}}
  {\mbox{$\mathbf\textstyle#1$}}
  {\mbox{$\mathbf\scriptstyle#1$}}
  {\mbox{$\mathbf\scriptscriptstyle#1$}}}

\title{Formal Verification of Station Keeping Maneuvers for a Planar Autonomous Hybrid System} 
\def\titlerunning{Formal Verification of Station Keeping}
\author{ 
Benjamin Martin
\institute{LIX, Ecole Polytechnique, CNRS\\ Universit{\'e} Paris-Saclay, 91120 Palaiseau,\\ France}
\email{bmartin@lix.polytechnique.fr}
\and
Khalil Ghorbal
\institute{INRIA, Rennes,\\ France}
\email{khalil.ghorbal@inria.fr} \vspace{1em}
\and
Eric Goubault
\institute{LIX, Ecole Polytechnique, CNRS\\ Universit{\'e} Paris-Saclay, 91120 Palaiseau,\\ France}
\email{goubault@lix.polytechnique.fr} 
\and
Sylvie Putot
\institute{LIX, Ecole Polytechnique, CNRS\\ Universit{\'e} Paris-Saclay, 91120 Palaiseau,\\ France}
\email{putot@lix.polytechnique.fr} 
}
\def\authorrunning{B. Martin et al.}

\newtheorem{proposition}{Proposition}
\newtheorem{definition}{Definition}
\newtheorem{theorem}{Theorem}
\newtheorem{lemma}{Lemma}
\newtheorem{remark}{Remark}

\begin{document}
%

%
%

\maketitle
\begin{abstract}
We formally verify a hybrid control law designed to perform a \emph{station keeping} maneuver for 
a planar vehicle. 
Such maneuver requires that the vehicle reaches a neighborhood of its station in finite time and remains in it while waiting 
for further instructions.   
We model the dynamics as well as the control law as a hybrid program and  
formally verify both the reachability and safety properties involved. 
We highlight in particular the automated generation of invariant regions which turns out to be crucial 
in performing such verification. 
We use the theorem prover \keyx{} to discharge some of the generated proof obligations. 

\end{abstract}


\section{Introduction}

Formal hybrid modelling languages such as hybrid automata \cite{Alur1993} or hybrid programs \cite{DBLP:journals/jar/Platzer08} offer a convenient way 
to describe a wide variety of hybrid systems. 
In this paper, we consider a piecewise continuous system where the continuous dynamics are subject to discrete switching. 
The plant part is modeled as a Dubins vehicle, that is a vehicle describing planar circular curves at a constant speed. 
The heading of the vehicle respects a hybrid control law, here taken from~\cite{Jaulin2013}, designed for station keeping maneuver. 
This means that the vehicle is expected to reach a neighborhood of its station in finite time and remains in it 
as long as it is not asked to do differently. 
Possible applications for such maneuvers are e.g. an autonomous sailboat that needs to anchor around a GPS position waiting to be picked up, or an autonomous drone that needs to keep a station at a given position while waiting for further instructions. 

Our goal is to formally verify the given control law while investigating to which extent such verification could be automated. 
In particular, 
%
we use recent symbolic computation techniques to automatically generate algebraic and semialgebraic invariant 
regions \cite{DBLP:conf/sas/MatringeMR10,Goubault2014}. 
Such regions are then exploited to formally verify the reachability and safety properties of the station keeping maneuver using 
the hybrid theorem prover \keyx~\cite{DBLP:conf/cade/FultonMQVP15}. 
We compare our findings with the results from~\cite{Jaulin2013} where another proof is conducted by means of guaranteed numerical methods.


\section{The Station Keeping Maneuver}

The Dubins vehicle in Cartesian coordinates is described by the following system:
\begin{equation}
    \label{eq:dubcarcart}
    \left\{ 
        \begin{array}{rcl}
            \dot{x} & = & \cos(\theta) \\
            \dot{y} & = & \sin(\theta) \\
            \dot{\theta} & = & u 
        \end{array}
    \right. ,
\end{equation}
where $(x,y,\theta)$ defines the pose of the vehicle composed of its position in the plane $(x,y)$ as well as its 
heading angle $\theta$. 
The vehicle is always moving at a constant speed (here fixed to $1$). 
The variable $u$ encodes an input control that affects directly the heading's angular velocity.
Following~\cite{Jaulin2013}, we consider the above plant model in the polar coordinates $(d,\varphi,\alpha)$ depicted in Figure \ref{fig:topolar}. The above ODE then becomes:  
\begin{equation}
    \label{eq:dubcarpol}
    \left\{
        \begin{array}{rcl}
            \dot{d} & = & -\cos(\cphi) \\
            \dot{\cphi} & = & \frac{\sin(\cphi)}{d}+u \\
            \dot{\alpha} & = & -\frac{\sin(\cphi)}{d}
        \end{array}
    \right. ,
\end{equation}
where radius $d$ is a positive real satisfying $d^2=x^2+y^2$ and the heading angle $\theta$ is linearly 
related to $\cphi$ and $\alpha$:  $\cphi - \theta + \alpha = \pi$. 
The angle $\cphi$ can be understood as a bearing which measures the angle of the head of the vehicle with respect to a given position of the plan (here the origin).

\begin{figure}
	\centering
	\includegraphics[width = 8.4cm]{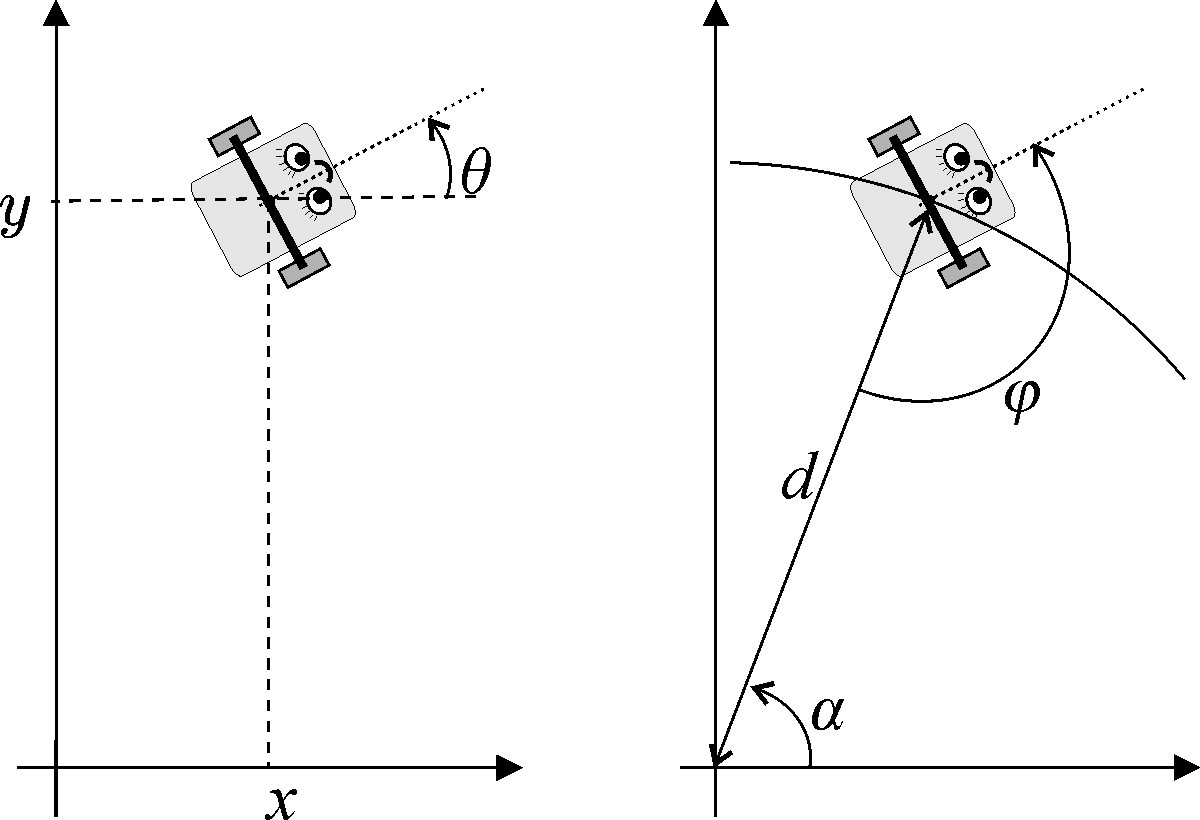}
    \caption{Cartesian coordinates (left) and polar coordinates (right). Courtesy to~\cite{Jaulin2013}.}
	\label{fig:topolar}
\end{figure}

The polar coordinates have in fact numerous advantages over the Cartesian coordinates. 
On one hand, one gets a decoupling of the state variables for free since the derivatives of $\cphi$ and $d$ are 
independent of $\alpha$ calling for a model reduction where only the states $(d,\cphi)$ are considered. 
On the other hand, since $\cphi$ appears only as a direct argument of the sine and cosine functions, one can 
restrict $\cphi$ to $[0,2\pi)$ with no loss of generality: the vector field is invariant under the action 
of the transitive additive group that takes $\cphi$ to $\cphi + 2k\pi$. 
Recall that the polar coordinates transformation presents a singularity at the origin $(x,y)=(0,0)$, where $d$ vanishes. 
We will go back to this issue later. 
First, we recall the piecewise control law for $u$ proposed in~\cite{Jaulin2013}.

\begin{equation}
\label{eq:claw}
		u=\left\{\begin{array}{ll}
                1 & \quad\mbox{if $\cos(\cphi)\leq \frac{\sqrt{2}}{2}$} \\
                -\sin(\cphi) &\quad \mbox{otherwise}
\end{array}\right. ,
		\end{equation}

The intuition is that the vehicle constantly turns left when it is not pointing towards the origin 
(modelled by $\cos(\cphi) \leq \frac{\sqrt{2}}{2}$), 
otherwise the input is proportional to the bearing $\cphi$ so as to push it towards 0, in which case the vehicle is moving towards the origin.

Combining the control law \eqref{eq:claw} together with the plant in polar coordinates \eqref{eq:dubcarpol}, 
one gets a switched system where two different dynamics, implied by two different controls, can be applied depending 
on the state of the system. 

%

In the following sections we prove that (i) the vehicle reaches in a finite time a position at a reasonably 
short distance from a beacon positioned at the center of the coordinate system, and (ii) stays in that region for 
an indefinite time. 
To do so, we first perform a qualitative analysis of the two continuous dynamics obtained exhibiting 
interesting invariant regions.

\section{Generating Positive Algebraic Invariants}
\label{algebraicinvariants}
For convenience, we recall the formal definition of positive invariant sets.
Let $\phi(\vec{x}_0,.): \R \rightarrow \R^n$ denote the solution of the initial value problem 
$\dot{\vec{x}}=\vec{f}(\vec{x})$ for a given ODE $\vec{f}$ and initial value $\vec{x}_0$. 
Let $I \subset \R$ denote the maximal interval on which $\phi(\vec{x}_0,.)$ is defined. 
Recall that $I$ need not be the entire real line and that, depending on $\vec{f}$, the solution may be defined on a bounded interval. 
When $I$ is bounded, the system exhibits a finite time blow-up problem \cite{Ball1978189}, that is in general at least one variable diverges in finite time. 
Such problems are intimately related to the singularities of the solutions and are often hard to detect and characterize. 
We will carefully discuss and analyze such issues for our case study. 
Notice, however, that to the best of our knowledge there are currently no automated methods to detect whether $I$ is bounded or not 
where non-linear dynamics are involved (when $\vec{f}$ is linear, $I=\R$). 

\begin{definition}
\label{inv}
A set $S \subseteq \R^n$ is \emph{positive invariant} for $\vec{f}$ if and only if for any $\vrx_0 \in S$, the 
corresponding solution $\phi(\vrx_0,\cdot)$ satisfies $\phi(\vrx_0,t) \in S$ for all $t \in [0,+\infty) \cap I$, that is 
for all non-negative time $t$ as long as the solution is defined. 
\end{definition}
In order to verify that a set $S$ is positive invariant for an ODE, it is enough to show that the flow $\vec{f}$ is entering, constant 
or inner tangential on the boundaries of $S$. 
When $S$ is semi-algebraic (that is defined by boolean combinations of polynomial equalities and inequalities), this 
can be done by checking the sign of the Lie derivatives (and, if necessary, higher-order Lie derivatives) 
of the active boundaries of $S$ \cite{DBLP:conf/emsoft/LiuZZ11}.

Recently in~\cite{DBLP:conf/sas/MatringeMR10,Ghorbal2014,DBLP:conf/emsoft/LiuZZ11,DBLP:conf/hybrid/Sankaranarayanan10,Goubault2014}, 
many effective methods for constructing algebraic and semi-algebraic positive invariant sets have been proposed. 
Those methods apply, however, only to polynomial vector fields. 
We thus start by transforming the system \eqref{eq:dubcarpol} into a polynomial differential system. 
This could be done classically by adding fresh variables corresponding to the transcendental functions. 
The so obtained dynamics is a sound approximation (abstraction) of the original dynamics that one could refine by re-introducing 
the links between the extra variables and the hidden transcendental functions they represent~\cite{Carothers}. 
For our dynamics, one obtains the following algebraic system~:
\begin{equation}
\label{eq:syspoly}
\left\{\begin{array}{rcl}
\dot{g} & = & -(he+u)h \\
\dot{h} & = & (he+u)g \\
\dot{e} & = & ge^2 \\ 
\dot{d} & = & -g \\
\dot{\cphi}& =& he+u
\end{array}\right. ,
\end{equation}
where the variable $g$ encodes cosine function $\cos(\cphi)$, $h$ the sine function $\sin(\cphi)$ and $e$ the inverse $\tfrac{1}{d}$. 
For the abstraction to be precise enough, one has to respect those functions initially, for instance if 
the initial values of $g$ and $h$ are fixed, then the initial value of $\phi$ is entirely determined. 
Likewise, if $d$ is fixed initially, then the initial value of $e$ is fixed and is equal to the inverse of $d$. 
In this case study, the control law $u$, as well as the switching conditions (cf. \eqref{eq:claw}), are expressible directly 
with the extra variables $g$, $h$, and $e$. Therefore, the control law as well as the plant could be rewritten entirely algebraically.  


\paragraph{Darboux polynomials and positive invariants}

Darboux polynomials, and more generally the Darboux criterion, are fundamental building blocks for the qualitative analysis 
of ODE and hence invariant generation. 
They are also at the heart of symbolic integration methods \cite{goriely2001integrability,Man93a}. 

\begin{definition}
\label{Darboux}
Let $\dot{\vec{x}}=\vec{f}(\vec{x})$ denote an ODE. 
A polynomial $p$ is \emph{Darboux} for $\vec{f}$
if and only if 
$$\dot{p}= c p$$
where $c \in \reals[\vec{x}]$ is a polynomial, called cofactor and 
where $\dot{p}$ denotes the time derivative of the polynomial $p$ with respect to $\vec{f}$. 
\end{definition}

The polynomial $\dot{p}$ is also known as the \emph{Lie derivative} of $p$ with respect to $\vec{f}$ and is formally defined as 
\[
    \left\langle \nabla p,\vec{f}\right\rangle,
\]
where $\nabla p$ is the gradient of $p$ and $\langle .,. \rangle$ is the standard inner
product on $\R^n$.

Searching for Darboux polynomials, up to a given fixed degree, can be performed algorithmically \cite{DBLP:conf/sas/MatringeMR10,Ghorbal2014} by 
deriving all the constraints that the unknown coefficients of a   polynomial (or template) have to satisfy and then solve the so obtained system.  

We were able to exploit such techniques to automatically generate Darboux polynomials of \rref{eq:syspoly}. 
For a better performance, we restricted ourselves to the three dimensional dynamics in $(g,h,e)$ since they define a closed 
form ODE on their own. 
\rref{tab:dbx} summarizes our findings depending on the selected control. 
From there, we recovered other Darboux polynomials for the five dimensional system by exploiting the algebraic invariant $de=1$ 
known to be satisfied by construction. 

\begin{table}
    \centering
    \begin{tabular}{c@{\hskip .5in}c@{\hskip .5in}c}
        Control & Darboux Polynomial & Cofactor \\
        \hline 
        \multirow{ 2}{*}{$u=1$}       & $e$           & $g e$     \\
                                      & $1+2 e h$     & $2 g e$   \\
        \hline
        \multirow{ 2}{*}{$u=-h$}      & $e$           & $ge$      \\
                                      & $h$           & $(e-1)g$  \\
    \end{tabular}
    \caption{Darboux polynomials for \rref{eq:syspoly}.}
    \label{tab:dbx}
\end{table}

From \rref{tab:dbx}, in the constant control mode ($u=1$), we observe that the cofactors of the 
two Darboux polynomials are the same up to a multiplication by the integer $2$. 
This suggests the following rational invariant function for this mode: 
\begin{equation}
	\label{eq:vcst}
	\Vcst := \frac{1+2 e h}{e^2} ,
\end{equation}
obtained by first matching the two cofactors by raising the power of the Darboux polynomial $e$ to match 
the multiplicative integer $2$ (since the cofactor of $e^2$ is twice the cofactor of $e$)  
and then dividing the two Darboux polynomials with the same cofactor, namely $1+2eh$ and $e^2$. 
Such relation between Darboux polynomials and rational invariant functions is well known in the literature \cite{goriely2001integrability}. 
One can indeed easily check that the Lie derivative of $\Vcst$ vanishes for all $t$, that is $\dVcst = 0$, as long as the control 
input remains equal to $1$. 
Moreover, since $de=1$ by construction of $e$, we have a polynomial equivalent formula for $\Vcst$, namely $d^2 + 2 d h$.

For the proportional mode, when $u=-h$, it turns out that the time derivative of $\Vcst$ keeps a constant sign:
\[
    \dVcst = -\frac{2g(h+1)}{e} \leq 0 \enspace .
\]
This is because $h+1 \geq 0$ (recall that $h$ is defined as a sine function) and $g > \tfrac{\sqrt{2}}{2}$ by definition of $g$ 
and the considered control law \eqref{eq:claw}. 

Figure \ref{fig:propctl} shows the region $\Vcst \leq 0$ (depicted in blue).  
Observe in particular that it contains the equilibrium of the system $(d,\cphi)=(1,\tfrac{3\pi}{2})$ and that $d$ 
is upper bounded by $2$, meaning that the vehicle is at a fairly close distance to the origin of the Cartesian coordinate system.  
Thus, the set $\Vcst \leq 0$ seems to be a good station keeping candidate for the switched system. 
\begin{figure}[tb]
    \centering
    \includegraphics[width=0.65\textwidth]{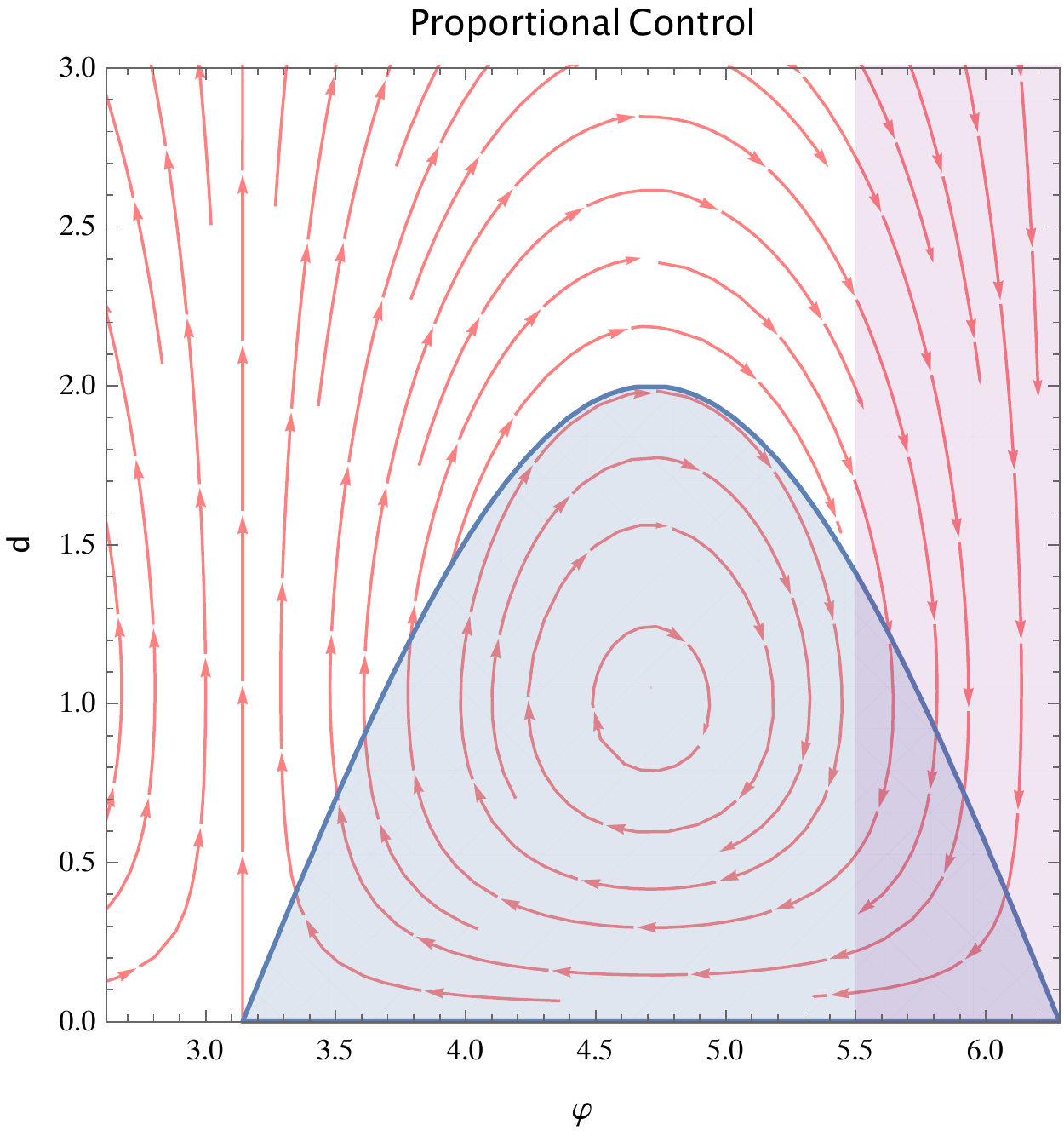}
    \caption{The set $\Vcst \leq 0$ is depicted in blue.} 
    \label{fig:propctl}
\end{figure}

In the next sections, we formally prove that $\Vcst \leq 0$ is an invariant set for the switched system and 
more importantly that it is reachable from any initial condition of the vehicle provided that initially $d>0$ and $\cphi \neq 0$. 

\section{Safety Analysis}

In this section, we will first model the switched system as a hybrid program, then prove the invariance of $\Vcst \leq 0$. 

\begin{definition}[Hybrid Program Model]
    \label{defn:hybridpgr}
\begin{equation*}
\alpha :=  
\left\{
 \left\{\operatorname{Plant}_{\mid u=1}\;\&\; d>0 \land 2g\leq \sqrt{2} \right\}\; \cup\;  \left\{\operatorname{Plant}_{\mid u=-h}\;\&\; d>0 \land 2g>\sqrt{2}\right\}
 \right\}^*
\end{equation*}
where the $\operatorname{Plant}$ dynamics are defined as in \eqref{eq:syspoly}. 
\end{definition}



The hybrid program $\alpha$ in Definition \ref{defn:hybridpgr} shows a piece-wise continuous system that 
models the behavior of the vehicle when the control law is applied. 
The entire feedback loop runs for any non-negative number of iterations, modeled by the star $\{\}^*$. 
The loop is made of an non-deterministic choice (modelled by the operator $\cup$) between the dynamics induced by the two possible controls ($\operatorname{Plant}_{\mid u=1}$ and $\operatorname{Plant}_{\mid u=-h}$). Any dynamics are applied as long as the states remains within the evolution domain given by the conditions after the \& symbol. Here, the different conditions on $g$ imposes to follow the control law given by \eqref{eq:claw}.
The condition $d>0$, present in both evolution domains, ensures that the polar coordinates rewriting is valid. 

In the sequel, we will be using $\Delta := \left( g^2+h^2=1 \land e d = 1 \land d>0 \right)$ to encode the fact that the 
initial value of the variables $(g,h,\cphi,d,e)$ is coherent, that is 
once $g$ and $h$ are fixed such that $g^2+h^2=1$, $\cphi \in [0,2\pi[$ is known and is such that $\cos(\cphi)=g$ and $\sin(\cphi)=h$, although its value is not given explicitly. 
The variable $e$ is entirely determined via the equation $d e = 1$ as soon as a positive $d$ is chosen. 
%
We are now ready to formally state the positive invariance of the region $\Vcst \leq 0$. 
\begin{theorem}[$\Vcst \leq 0$ is a Positive Invariant]
\label{thm:safety}
\[
\Vcst \leq 0 \; \land \; \Delta
\rightarrow
\;[ \alpha ] \ \Vcst \leq 0
\]
\end{theorem}
%
The box modality around $\alpha$ means that for all runs of the hybrid program, the following post-condition must be true. When \rref{thm:safety} is written as a hybrid program in \keyx{}\footnote{Source files for this hybrid program are available at the following link \url{http://ben-martin.fr/publications}}, we use the polynomial form of $\Vcst = d^2 + 2dh$, which is a valid rewriting of the rational form if $de-1=0$ is an invariant of the hybrid system.
If one assumes that $de-1= 0$ holds initially (as this is the case here) then it is possible to prove that it holds for all time. Indeed the polynomial $d e - 1$ is a Darboux polynomial for the dynamics defined in \eqref{eq:syspoly} for all inputs $u$. This fact cannot be proved currently within the theorem prover \keyx{} as a proof rule based on the Darboux criterion is not yet available. 
Therefore, it is currently necessary to add the conditions $de-1=0$ into the evolution domain in order to complete the proof in \keyx{}. 
The proof itself then is mostly based on the differential invariant (DI) proof rule~\cite{DBLP:journals/jar/Platzer08}, which 
is essentially a conservative lifting of barrier certificates~\cite{Prajna+Jadbabaie/2004/Safety} to the boolean connectives. 
For convenience, we give in \rref{eq:DI} the conditions required for barrier certificates, or likewise the premises of the proof rule DI for 
the simple case of region of the form $p\leq0$. 
%
%
\begin{equation}
    \label{eq:DI}
    (DI)\frac{\forall x. ([x':=f(x)]\dot{p} \leq 0)}{p\leq0 \rightarrow [\dot{x}=f(x) \& H] p\leq 0}
\end{equation}
In other words, the invariance of $p\leq 0$ can be deduced if its time (Lie) derivative is negative.
The reason (DI) succeeds is due to the fact that when $g> \frac{\sqrt{2}}{{2}}$ and the proportional control is applied, one gets 
$\dVcst = -2gd(h+1)$, which is negative since $g$ is positive, $d$ is positive, and $h$, as a cosine, is greater or equal to $-1$.
%
The set $\Vcst \leq 0$ is thus not only a positive invariant when the constant control 
is applied but also for the proportional control when applied according to the control law \rref{eq:claw}. 
The output obtained with \keyx{} is shown on~\rref{fig:keyF}.

\begin{figure}[tb]
	\centering
	\includegraphics[width=0.65\textwidth]{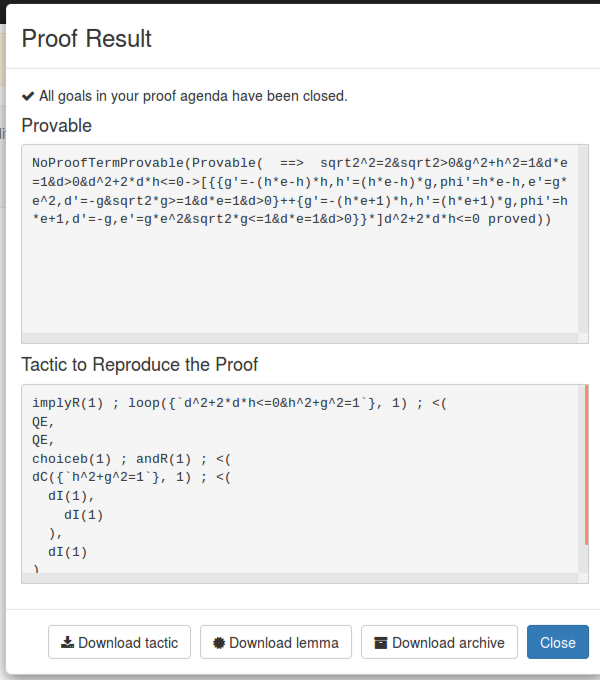}
	\caption{\keyx{}: last window showing the proof of \rref{thm:safety} is concluded.}
	\label{fig:keyF}
\end{figure}

\section{Reachability Analysis}
\label{sec:reachability}

We prove in this section that the region $\Vcst \leq 0$ is reachable from (almost) anywhere in the phase space. 
We use the same hybrid program model given in Definition \ref{defn:hybridpgr}.  
We first discuss the special case $\cphi=0$ (\rref{sec:special}), and then the generic case where $\cphi > 0$ (\rref{sec:generic}). 

\subsection{The Special Case (\texorpdfstring{$\cphi=0$}{Singular})} 
\label{sec:special}
When $\cphi=0$ initially, then $g$ and $h$ have to be instantiated to $1$ and $0$ respectively. 
This configuration sets the control input $u$ to $-h$ and one can see from Table \ref{tab:dbx} that $h$ is a Darboux polynomial 
and therefore $h=0$ is an invariant equation as long as the system evolves following the $\operatorname{Plant}$ equations (cf. \eqref{eq:syspoly}) with the domain $d>0$. 
In this particular case, and since $h=0$ is invariant, the set of equations simplifies to
\begin{equation}
    \label{eq:syspolySpecial}
    \left\{\begin{array}{rcl}
            \dot{g} & = & 0 \\
            \dot{h} & = & 0 \\
            \dot{e} & = & e^2 \\ 
            \dot{d} & = & -1 \\
            \dot{\cphi}& =& 0
        \end{array}\right. ,
    \end{equation}
meaning that the bearing angle $\cphi$, as well as its cosine $g$, will remain constants. 
Subsequently, this also means that the control input $u$ will also remain fixed to $-h$, that is $0$. 
The vehicle shows here an interesting behavior as the time derivative for $d$ is strictly decreasing: 
the vehicle is heading straight to the origin $(x,y)=(0,0)$. 
However, because of the evolution domain constraint $d>0$, the dynamics will be followed as long as this constraint is not violated. 
But then, by design, the system is forced again to execute the differential equation which only makes $d$ closer and closer to $0$ 
(and thus $e$ closer to infinity because of $e d=1$ is an invariant). 
We have here in fact a finite time explosion problem: if the system starts at a distance $d_0$ from the origin, with $\cphi_0=0$, the dynamics are only defined for $t \in [0,d_0[$ and the maximal interval of definition is upper bounded.  
At $t=d_0$, the model hits the singularity of the polar coordinates transformation and it is no longer valid as is. 
A careful analysis shows that, right after the singularity, $d$ remains an infinitesimal (and is thus continuous as one expects) and $\cphi$ is discontinuous 
as it jumps from $0$ to $\pi$ switching the control input from $-h$ to $1$. 
This discontinuity comes in fact from switching the direction of the radius vector (of magnitude $d$) and does affect neither the position nor the heading of the vehicle. 
The vehicle follows therefore a new trajectory with $d$ very small but positive and $\cphi = \pi$. 
This new initial position is part of the generic case discussed in the next section.


\subsection{The Generic Case (\texorpdfstring{$\cphi>0$}{Non Singular})} 
\label{sec:generic}
The phase space is now restricted to $(0 < \cphi < 2\pi)$. 
We assume that in this case all trajectories are defined for all $t \geq0$. 
%
%
Our reachability analysis exploits the recent invariant-based liveness proof rule, $(SP)$, introduced by Sogokon 
and Jackson in~\cite[Proposition 10]{Sogokon2015}. 
The idea is to use special invariant sets, so called \emph{staging sets}, that contains the initial set and from which the system can only escape to go to the target set one wants to prove reachable. 
To further prove that the system will eventually leave the staging set in finite time, a progress function must be also supplied. 
These two ingredients are sufficient to prove that \emph{any} trajectory starting in the staging set will eventually reach the target set in finite time. 
The $(SP)$ rule is defined as follows: $X_T \subset \R^n$ is the target set we want to prove reachable in finite time and $X_0 \subset \R^n$ is the initial set. 
The premises of the proof rule are sufficient to prove its conclusion under the assumption that the solution is defined for as long as needed to reach the target set. 
\begin{equation}
    \label{eq:EvenR}
    (\operatorname{SP})\frac{\begin{array}{c@{\hskip 20pt}c@{\hskip 20pt}c}\multicolumn{3}{c}{\vdash \exists \epsilon > 0.\ \forall \vec{x}.\ S \rightarrow (p \geq 0 \land \dot p \leq -\epsilon)} \\[10pt] 
    X_0 \land \lnot X_T \vdash S & \vdash S \rightarrow [\dot{\vec{x}}=f(\vec{x})\; \&\; \lnot (H \land X_T)] S & X_0 \lor S \vdash H\\[2pt] \end{array}}{\vdash X_0 \rightarrow \langle \dot{\vec{x}}=f(\vec{x})\; \&\; H \rangle X_T}
\end{equation}
The diamond modality around the hybrid program means that at least one run leads to satisfying the post-conditions (liveliness).
The proof rule has four premises and relies essentially on a strictly decreasing progress function $p$ within an invariant set $S$. 
The progress of $p$ is proven using the positive real number $\epsilon$. 
The intuition of the proof rule $(\operatorname{SP})$ is that if there exists a bounded from below and decreasing function along the trajectories in an 
invariant region with respect to certain evolution domain constraint, then the flow cannot stay indefinitely inside and must eventually exit the region described by those constraints.

The aim of this section is to prove that the region defined by $\Vcst \leq 0$ is reachable from any generic initial position (i.e. excluding the case $\cphi=0$). 
Doing so with a single application of the $(SP)$ rule can be cumbersome. Therefore, we suggest to build a chain of staging sets that prove reachability of intermediate regions, the chain leading to the reachability of $\Vcst \leq 0$.
In what follows, we partition the phase space into $4$ regions as shown in \rref{fig:phasespace}: 
\begin{equation}
    \begin{split}
    \text{\ding{192}} &:= 0 < \cphi < \frac{\pi}{4} \;\land\; d>0 \\
    \text{\ding{193}} &:= \frac{\pi}{4} \leq \cphi \leq \frac{7\pi}{4} \;\land\; d>0 \;\land\; \Vcst > 0 \\ 
    \text{\ding{194}} &:= \frac{7\pi}{4} < \cphi < 2\pi \;\land\; d>0 \;\land\; \Vcst > 0 \\
    \text{\ding{195}} &:= \Vcst \leq 0
    \end{split}
    \label{eq:regions}
\end{equation}
\begin{figure}[tb]
    \centering
    \includegraphics[width=0.65\textwidth]{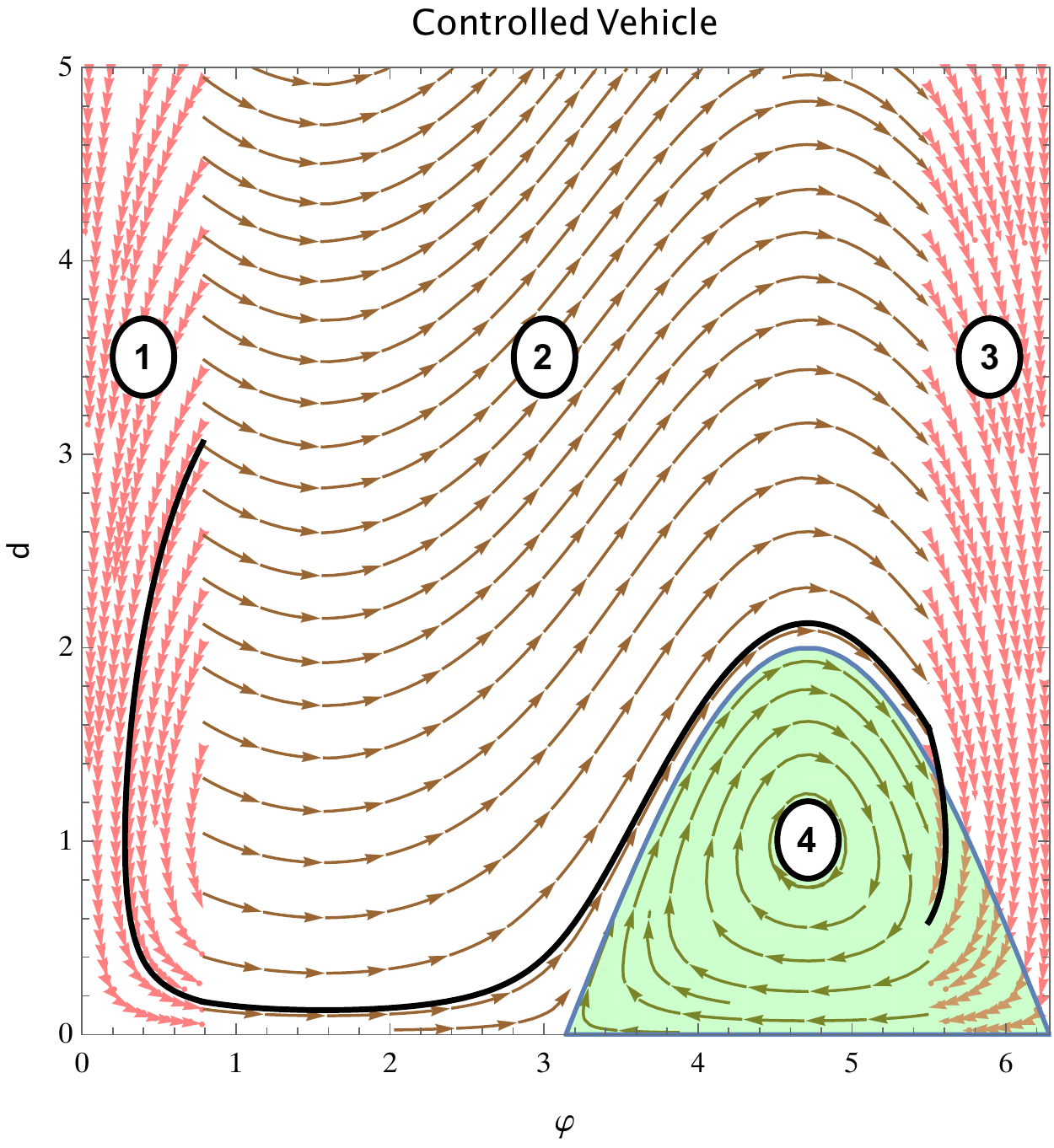}
    \caption{Phase space of the controlled Dubins vehicle in polar coordinates. 
        The black curve shows a possible trajectory of the vehicle until reaching 
    $\Vcst \leq 0$ (region \ding{195}). When the vehicle is in regions \ding{192} or \ding{194}, $u=1$; In region \ding{193} $u=-h$.} 
    \label{fig:phasespace}
\end{figure}
This partition has been found manually by separating the regions with different control and the region found safe in the previous section. We prove sequentially that $\text{\ding{193}}$ can be reached from $\text{\ding{192}}$, $\text{\ding{194}}$ can be reached from $\text{\ding{193}}$ and that $\text{\ding{195}}$ can be reached from $\text{\ding{194}}$.

\begin{lemma}
    \label{lem:reach1}
    \[
        \left(\frac{\sqrt{2}}{2} < g \;\land\; h > 0 \;\land\; \Delta \right) \rightarrow \langle\;\operatorname{Plant}_{\mid u=-h}\;\rangle\
        \left(g=\frac{\sqrt{2}}{2} \;\land\; h=\frac{\sqrt{2}}{2} \right)
    \]
\end{lemma}

\begin{proof}
    Apply the proof rule $(\operatorname{SP})$ with the progress function $p := d$, $\epsilon := \frac{\sqrt{2}}{2}$ and the invariant set 
    $S := X_0$. All premises can be checked automatically. The first, second and fourth premises can be discharged 
    using a quantifier elimination procedure over the reals (CAD for instance). The most involved premise is the third one where 
    one has to prove the invariance of the staging set $S$. As shown in \cite{DBLP:conf/emsoft/LiuZZ11} this can be also 
    reduced to a universal quantifier elimination problem and can thus be discharged using CAD. 
\end{proof}

\rref{prop:reach1} exhibits a particular run of the hybrid program $\alpha$ in which it reaches the bearing 
angle $\cphi=\tfrac{\pi}{4}$ from 
any coherent initial position satisfying $\cphi \in (0,\tfrac{\pi}{4})$. 
\begin{proposition}
    \label{prop:reach1}
    \[
        \left(0 < \cphi < \frac{\pi}{4} \;\land\; \Delta \right) \rightarrow \langle \alpha \rangle\ \left(\cphi=\frac{\pi}{4}\right) 
    \]
\end{proposition}
\begin{proof}
    When $\cphi \in (0,\tfrac{\pi}{4})$ initially then $g > \tfrac{\sqrt{2}}{2}$ and the control input $u$ is set to $-h$. 
    According to \rref{lem:reach1}, and since $d>0$ is a positive invariant, $\cphi=\tfrac{\pi}{4}$ is reachable by 
    continuously following the dynamics while fixing $u$ to $-h$. 
    This is allowed by the control law since, as long as the system evolves within the staging set, $g > \tfrac{\sqrt{2}}{2}$ is 
    satisfied.  
\end{proof}

Similarly, we prove that region \ding{194} is reachable from \ding{193} by exhibiting a run of $\alpha$ that reaches the region \ding{193}.  
\begin{lemma}
    \label{lem:reach2}
    \begin{equation*}
        \left(g \leq \frac{\sqrt{2}}{2} \;\land\; \Vcst>0 \;\land\; \Delta \right) 
        \rightarrow \langle\;\operatorname{Plant}_{\mid u=1} \;\rangle \left(\frac{\sqrt{2}}{2} < g \;\land\; h < 0 \right) 
    \end{equation*}
\end{lemma}
\begin{proof}
    Apply the proof rule $(\operatorname{SP})$ with the progress function $p := -\cphi+\tfrac{7\pi}{4}$, $\epsilon := \frac{1}{2}$ and the invariant set 
    $S := X_0$. To prove that $S \rightarrow p\geq 0$, we need to use the fact that $(g,h)=(\cos(\cphi),\sin(\cphi))$. 
    The proof is not involved but has to be done manually as it requires manipulating transcendental functions. 
\end{proof}

\begin{proposition}
    \label{prop:reach2}
    \[
        \left(\frac{\pi}{4} \leq \cphi \leq \frac{7\pi}{4} \;\land\; \Vcst > 0 \;\land\; \Delta \right) \rightarrow \langle \alpha \rangle\ \left(\frac{7\pi}{4} < \cphi < 2\pi\right) 
    \]
\end{proposition}

Finally, we prove the reachability of the invariant set \ding{195} from \ding{194}. 
\begin{lemma}
    \label{lemma:reach3}
    \[
        \left(\frac{\sqrt{2}}{2} < g \;\land\; h < 0 \;\land\; \Vcst > 0 \;\land\; \Delta \right) \rightarrow \langle\;\operatorname{Plant}_{\mid u=-h} \;\rangle\ 
        \left(\Vcst \leq 0 \right)
    \]
\end{lemma}
\begin{proof}
    Apply the proof rule $(\operatorname{SP})$ with the progress function $p := d$, $\epsilon := \frac{\sqrt{2}}{2}$ and the invariant set 
    $S := X_0$. 
\end{proof}

\begin{proposition}
    \label{prop:reach3}
    \[
        \left(\frac{7\pi}{4} < \cphi < 2\pi \;\land\; \Vcst > 0 \;\land\; \Delta \right) \rightarrow \langle \alpha \rangle\ \left(\Vcst \leq 0\right) 
    \]
\end{proposition}

Combined together, Propositions \ref{prop:reach1}, \ref{prop:reach2} and \ref{prop:reach3} give:
\begin{theorem}[Reachability of $\Vcst$]
    \label{thm:reachability}
    The region $\Vcst \leq 0$ is reachable from any coherent generic position: 
    \[
        \left(0 < \cphi < 2\pi \;\land\; \Delta \right) \rightarrow \langle \alpha \rangle\ \left(\Vcst \leq 0\right) 
    \]
\end{theorem}

	The proof rule $(SP)$ is not yet available in \keyx{}. 
    All the premises can be however discharged using a quantifier elimination procedure, including the invariance of the staging 
    set $S$. The latter is a direct consequence of the decidability of the invariance of semialgebraic sets~\cite{DBLP:conf/emsoft/LiuZZ11}. 
    We used our proper implementation of the procedure described in~\cite{DBLP:conf/emsoft/LiuZZ11} together with the \texttt{Reduce} 
    procedure in Wolfram Mathematica to discharge all the premises of the $(SP)$ proof rule, except for \rref{lem:reach2} where transcendental functions are involved. 
%

\begin{remark}
\rref{thm:reachability} tells nothing about the time required to reach the target set, only that this time is finite. 
    The rule $(SP)$ embeds, however, a bounded progress function which can be used to determine an upper bound on the time 
    required to reach the target region. 
	For example starting at $(\cphi_0, d_0)$ in region \ding{192} ($d >0$), and given the upper bound on the decrease of the progress 
    function $d$, we can conclude that the region \ding{193} is entered in at most at $\sqrt{2} d_0$ seconds.
    The quality of this upper bound depends on the positive bound used for the progress function and can thus be arbitrarily large 
    (but always finite). 
    Notice that a finer analysis for the time spent at a given traversed region would benefit from a lower bound of the progress function, although such bound is not required for proving the reachability itself. 
\end{remark}

\begin{remark}
    There exists in fact an attractor for the switched system that is inside $\Vcst \leq 0$, namely $\Vcst\leq -\tfrac{1}{2}$. 
    However, proving its reachability is much more involved than proving 
the reachability of $\Vcst \leq 0$, because a trajectory may loop for some time before reaching it. 
    It also features a sliding mode at the boundary $\cphi=\tfrac{7\pi}{4}$ for $d$ in $\left[\tfrac{1}{\sqrt{2}},1\right]$. 
    In fact the only entry to the attractor is the point $\left(\tfrac{7\pi}{4},\tfrac{1}{\sqrt{2}}\right)$ which is reached whenever the system enters 
    its sliding mode. Notice also that when the system loops around this attractor, one has to consider roots of the Lambert W function. 
    We do not carry such proof in \keyx, however, since $\Vcst \leq 0$ is sufficient to prove that the maneuver reaches a region close to the origin. 
    (see \rref{fig:refined}). 
\end{remark}

\begin{figure}[bt]
    \centering
    \includegraphics[width=0.65\textwidth]{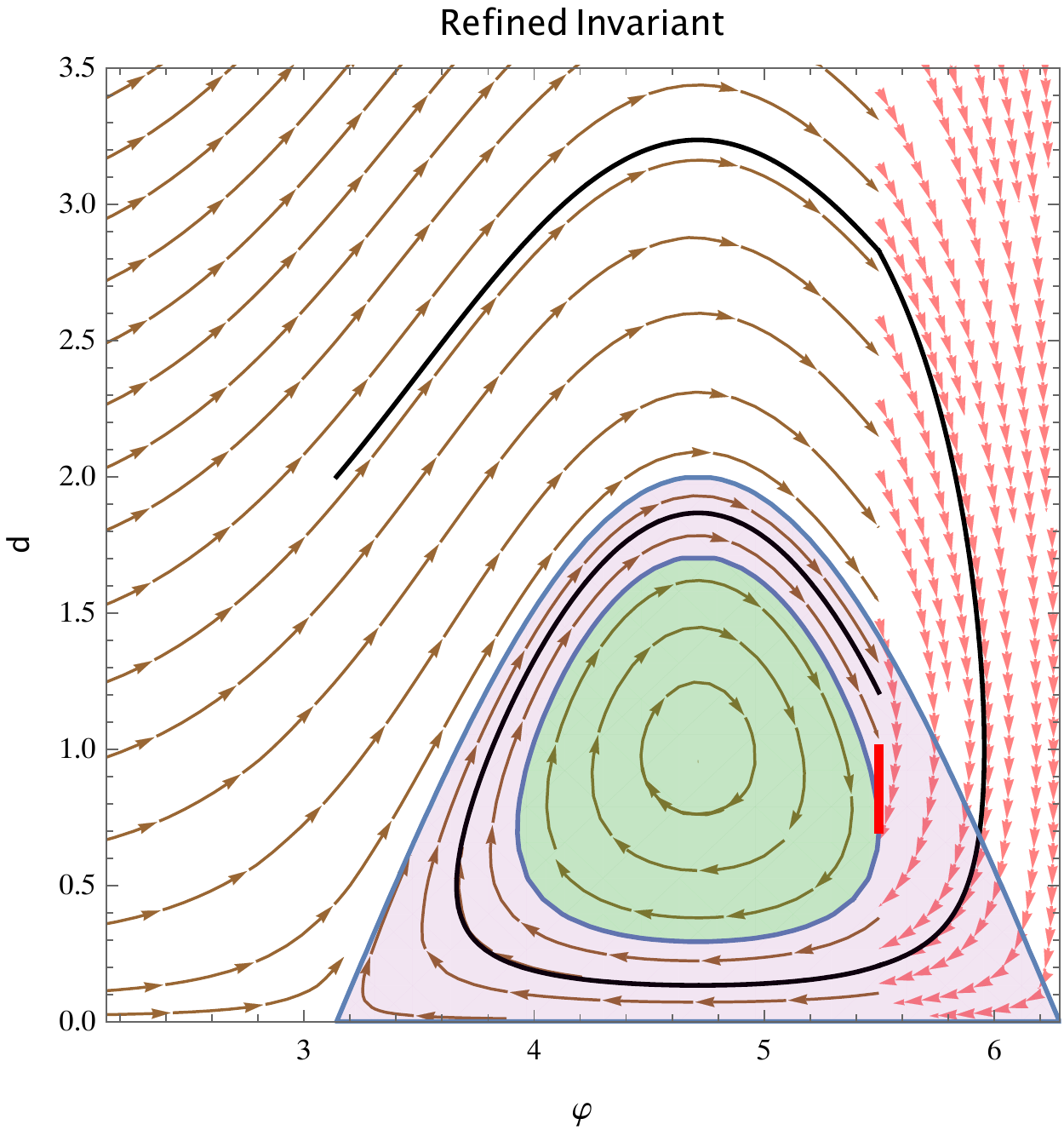}
    \caption{The refined invariant of the controlled vehicle. The red segment shows the sliding mode region before entering the attractor. 
    The black curve shows a possible looping trajectory around the attractor. The region $\invcst$ is shown in purple for convenience. }
    \label{fig:refined}
\end{figure}

\section{Related Work}


An interval-based numerical approach has been proposed in~\cite{Jaulin2013} to validate the controller.  
The idea is to construct a discrete abstraction of the state space of the system and then to build a graph where 
 nodes correspond to discrete regions and transitions between nodes mean that the system can potentially move from 
 one region to the other.
To actually build such transitions, the author considered guaranteed numerical tests based on interval analysis. 
This final graph was then used to show that 
the vehicle will be eventually trapped in a limited region of the state space. 
By construction, such a region is an over-approximation of the actual attractor of the system. 
%
%
This region is not precisely given in~\cite{Jaulin2013} but it can be approximated by the following set: 
\begin{align*}
&( 0 \leq d \leq 2) \; \vee \\
 &\left( \left( 0 \leq \cphi \leq \frac{\pi}{6} \vee \frac{7\pi}{4}\leq \cphi \leq 2 \pi \right) \land 0 \leq d \leq 7.6 \right) \;\vee \\
  &\left( \frac{\pi}{2} \leq \cphi \leq \frac{7\pi}{4} \land 2 \leq d \leq 7.6 \land \frac{112}{\pi}\cphi -25 d - 6 \geq 0 \right).
\end{align*}
In comparison, the zero level set of $\Vcst$ describes more accurately the behavior of the system, entailing a sharper analysis. 
\rref{fig:invariantsComp} depicts the attractor from~\cite{Jaulin2013} in comparison to the ones found in this paper. 
It is worth noting that the method from~\cite{Jaulin2013} does not require the algebraic rewriting. It works directly with 
the original system in polar coordinates and uses guaranteed numerical computations to find the invariant set. 
As in our approach, however, finding invariant candidates was essentially manual. 

\begin{figure}[htb]
	\centering
	\includegraphics[width= 0.65\textwidth]{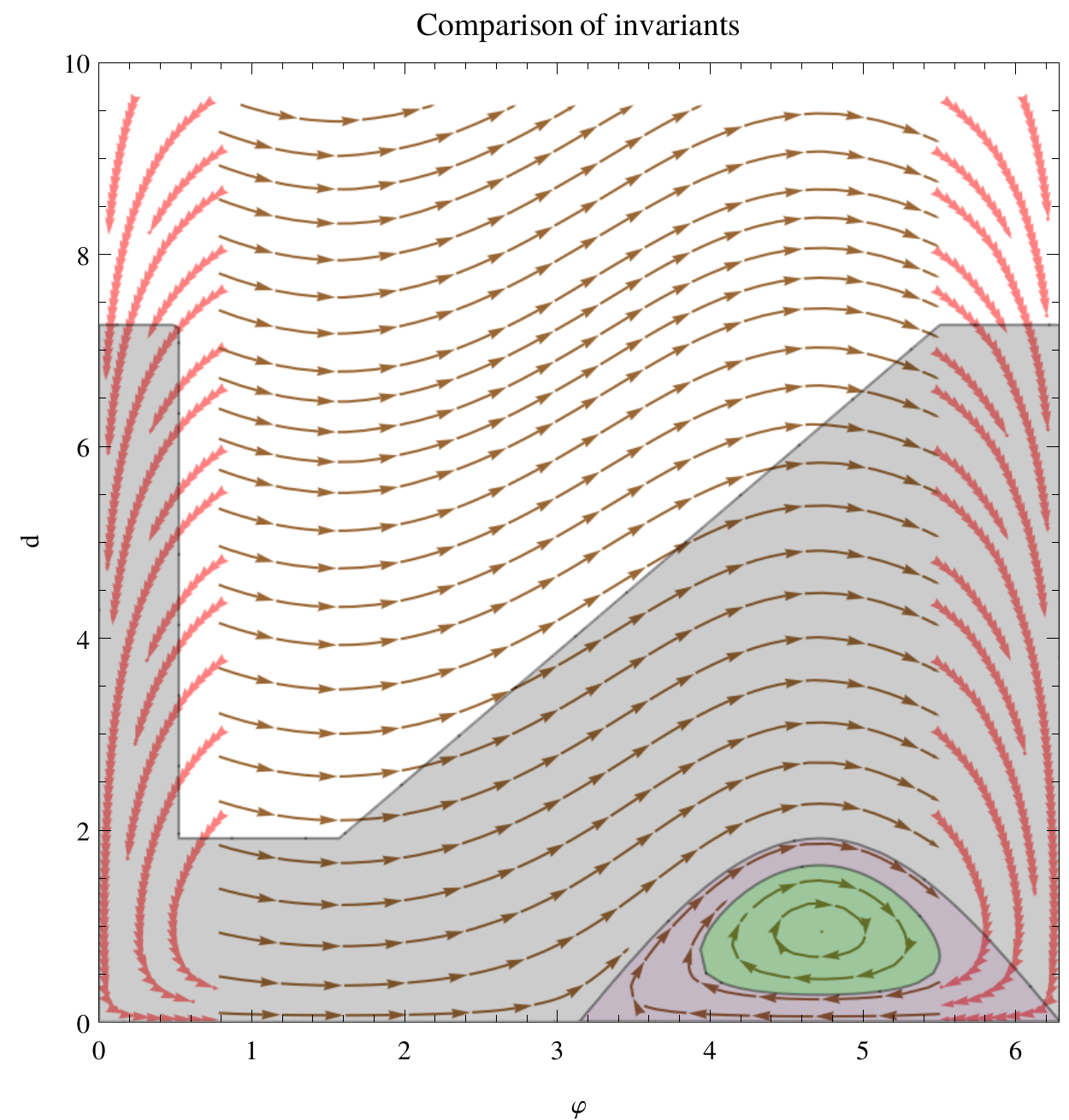}
	\caption{In gray the invariant region obtained in~\cite{Jaulin2013}, in purple $\Vcst \leq 0$ and in green $\Vcst \leq -\frac{1}{2}$.}
	
	\label{fig:invariantsComp}
\end{figure}

\section{Discussion and Conclusion}

We have developed a formal proof for the safety and liveness of an autonomous switched system, corresponding 
to a planar Dubins car whose goal is to perform a station keeping maneuver around the origin. 
This proof can be synthesized in three steps. 
First, we have used recent algebraic methods to derive algebraic invariant properties for the switched system. 
In particular, an algebraic invariant region, corresponding to a station keeping behavior, has been 
identified and formally checked with the hybrid system theorem prover \keyx . 
Second, the unbounded time reachability analysis of the system has been performed. 
To do so, the phase space was partitioned into subregions, each equipped 
with a progress function to prove that the system eventually leaves the region. 
To complete the proof, we exploited the invariants we generated to show that a region is only exited at certain locations. 

Although this case study appears simple, this formal proof relies on non-trivial elements in particular for the reachability analysis. 
We believe, however, that the current state-of-the-art techniques and tools are mature enough to handle such the proof obligations for such case study (up to properly implementing the used proof rules in a theorem prover like \keyx{}).
Some questions remain about how far can the application of these tools be automatized. For example, we were required to decompose the state space for the proof of reachability. The decomposition is obtained by hand from the autonomous system itself and the invariant found for the safety proof. We can think that for other applications, such a decomposition could be partially inferred automatically in a similar manner.

Another challenging future work avenue we are keen to investigate is the interaction between the formal proof process and the design of the control law. For instance, how can one derive a feedback for the designer when the proof fails ? 
This becomes more intricate when the proof strategy proceeds by sufficient conditions not by equivalences, as for reachability. 
A promising direction would be to direct the design so as to ease the generation of suitable differential variants/invariants 
sets (e.g. using Darboux polynomials), like for control Lyapunov functions.


%
%





\section*{Acknowlegments}

This work is partially supported by the DGA / MRIS under the project ``Suret\'e de Fonctionnement des Syst\`emes Robotiques Complexes" and by the Academic and Research Chair ``Engineering of Complex Industrial Systems", sponsored by Thales, Naval Group, Dassault Aviation and DGA.

\bibliographystyle{eptcs}
\bibliography{biblio}
\end{document}